\documentclass{llncs}
\usepackage{latexsym}
\usepackage{amsmath}
\usepackage{amssymb}
\usepackage{mathtools}
\usepackage[table]{xcolor}
\usepackage{graphicx}
\usepackage{verbatim}
\usepackage{xspace}
\usepackage{cite}
\usepackage{enumerate}
\usepackage{algorithm}
\usepackage[noend]{algpseudocode}
\floatname{algorithm}{Procedure}
\usepackage{placeins}

\let\doendproof\endproof
\renewcommand\endproof{~\hfill\qed\doendproof}

\usepackage{tikz}
\tikzstyle{defnode}=[circle,fill=black,draw=black,inner sep=0pt,minimum size=4pt]

\def\OPT{\texttt{OPT}\xspace}
\def\ALG{\texttt{ALG}\xspace}

\def\o{\texttt{O}\xspace}
\def\A{\texttt{A}\xspace}

\def\ec{{\sc edge-coloring}\xspace}

\newcommand{\ab}[1] {\left\vert #1\right\vert}
\DeclarePairedDelimiter{\ceil}{\lceil}{\rceil}

\DeclareMathOperator{\dgn}{dgn}

\title{Optimal Online Edge Coloring of Planar Graphs with Advice}
\author{Jesper W. Mikkelsen\thanks{Supported in part by the Villum Foundation and the Danish Council for Independent Research, Natural Sciences.}
\email{\\ jesperwm@imada.sdu.dk}
\institute{Department of Mathematics and Computer Science\\
University of Southern Denmark\\}}

\begin{document}
\maketitle

\begin{abstract} 
Using the framework of advice complexity, we study the amount of knowledge about the future that an online algorithm needs to color the edges of a graph optimally, i.e., using as few colors as possible. For graphs of maximum degree $\Delta$, it follows from Vizing's Theorem that $O(m\log \Delta)$ bits of advice suffice to achieve optimality, where $m$ is the number of edges. We show that for graphs of bounded degeneracy (a class of graphs including e.g. trees and planar graphs), only $O(m)$ bits of advice are needed to compute an optimal solution online, independently of how large $\Delta$ is. On the other hand, we show that $\Omega (m)$ bits of advice are necessary just to achieve a competitive ratio better than that of the best deterministic online algorithm without advice. Furthermore, we consider algorithms which use a fixed number of advice bits per edge (our algorithm for graphs of bounded degeneracy belongs to this class of algorithms). We show that for bipartite graphs, any such algorithm must use at least $\Omega(m\log \Delta)$ bits of advice to achieve optimality.
\end{abstract}

\section{Introduction}
An \emph{edge coloring} of a graph is an assignment of colors to the edges of the graph such that no two adjacent edges share the same color. Many scheduling and assignment problems can be modeled as edge coloring problems. The \emph{online edge coloring problem}, which we refer to simply as \ec, was introduced by Bar-Noy et~al.~\cite{edgecoloring}. In this problem, the edges of a graph arrive one by one. The edges are specified by their endpoints, but the vertices of the graph are not known in advance. Each edge must be assigned a color before the next edge arrives, under the constraint that no two adjacent edges are assigned the same color. The color assigned to an edge cannot be changed later on. The goal is to use as few colors as possible. 

Traditionally, worst-case competitive analysis \cite{CompRatio1, CompRatio2} is used to measure the performance of an online algorithm. The solution produced by the online algorithm, $\ALG$, is compared to that of an optimal offline algorithm, $\OPT$, which knows the entire input in advance. 
More precisely, let $\ALG(\sigma)$ ($\OPT(\sigma)$) denote the number of colors used by $\ALG$ ($\OPT$) when coloring a sequence, $\sigma$, of edges. We say that $\ALG$ is \emph{$c$-competitive} if there exists a constant $c_0$ such that $\ALG(\sigma)\leq c\cdot\OPT(\sigma)+c_0$ for any input sequence $\sigma$. If the inequality holds with $c_0=0$, then $\ALG$ is said to be \emph{strictly $c$-competitive}.

In \cite{edgecoloring} it is shown that any \ec algorithm, which never introduces a new color unless forced to do so, is strictly $2$-competitive and that no online algorithm, even if we allow randomization and restrict the input graph to being a forest, can achieve a better competitive ratio.

The underlying assumption of competitive analysis, that nothing is known about future parts of the input, is sometimes unrealistic. Therefore, for many online problems, various relaxations of this assumption have been suggested, including look-ahead~\cite{lookahead}, locality of reference~\cite{lor} and several models where the input is generated from some known probability distribution~\cite{beyondcomp, randor, stat}. 
In this paper, we consider the recent idea of \emph{advice complexity} introduced in \cite{A4} and further developed in \cite{A1, A2, A3}. Advice complexity provides a quantitative and problem-independent approach for relaxing the online constraint by providing the algorithm partial knowledge of the future. Our main goal in applying the framework of advice complexity to \ec is to better understand the online hardness of the problem. How much (and which kind of) information about the future are we lacking in order to produce an optimal edge coloring in the online setting?

\subsubsection{Advice complexity models.}  
In this paper, we consider the two most widely used models of advice complexity. In both models, an oracle, which has unlimited computational power and knows the entire input, provides the online algorithm $\ALG$ with some advice bits. For \ec, the input is a sequence of $m$ edges $\langle e_1,\ldots , e_m\rangle$. The two models are defined as follows:

\emph{Advice-with-request} \cite{A2}. In this model, \ALG receives some fixed number, $b$, of advice bits along with each request. 
 That is, when the edge $e_i$ arrives, the algorithm receives some advice $b_i\in\{0,1\}^b$ from the oracle. The algorithm then decides which color to assign to $e_i$ based on the edges $e_1,\ldots , e_i$ that have been revealed up until now and the advice $b_1,\ldots, b_i$ received thus far.  

\emph{Advice-on-tape} \cite{A1, A3}. In this model, the online algorithm $\ALG$ is provided access to an infinite advice tape prepared by the oracle. The algorithm may, at any point in time, read some number of advice bits from the tape. When the edge $e_i$ arrives, the algorithm must decide which color to assign to $e_i$ based on the edges $e_1,\ldots , e_i$ that have been revealed up until now and the advice read so far from the tape. The \emph{advice complexity}, $b(m)$, of \ALG is the largest number of advice bits read by \ALG, over all possible input graphs with at most $m$ edges.  

Note that an algorithm in the advice-with-request model receives exactly $bm$ bits of advice in total. Thus, it can be converted into an algorithm with advice complexity $bm+O(\log b)$ in the advice-on-tape model (the $O(\log b)$ bits of advice are used to encode $b$). Converting in the opposite direction is not always possible in a meaningful way. In particular, an algorithm in the advice-on-tape model is allowed to read only a sublinear number of advice bits. This is not possible in the advice-with-request model.

\subsubsection{Preliminaries.}
All graphs considered are simple. We denote the number of edges in a graph by $m$, the number of vertices by $n$ and the maximum degree by $\Delta$. 
A graph $G$ is \emph{$k$-edge-colorable} if there exists an edge coloring of $G$ with at most $k$ different colors. The \emph{chromatic index} $\chi ' (G)$ of $G$ is the smallest integer $k$ such that $G$ is $k$-edge-colorable.  We assume that colors are represented by consecutive positive integers. 
For a bipartite graph $G$, we write $G=(L,R)$ if $L$ and $R$ form a bipartition of the vertices of $G$. We let $K_{a,b}$ denote the complete bipartite graph $G=(L, R)$ where $\ab{L}=a$ and $\ab{R}=b$.

In addition to bipartite graphs, we consider trees, planar graphs and, more generally, $d$-degenerate graphs. A graph is \emph{$d$-degenerate} if there is an ordering $v_1, v_2,\ldots , v_n$ of its vertices such that, for $1\leq i\leq n$, the vertex $v_i$ is adjacent to at most $d$ vertices in $\{v_1,\ldots , v_{i-1}\}$. The \emph{degeneracy} of a graph $G$ is the least integer $d$ such that $G$ is $d$-degenerate. An edge $e=(v_i, v_k)$ where $i<k$ is said to be a \emph{front-edge} at $v_i$ and a \emph{back-edge} at $v_k$. Furthermore, $d_f(v_i)$ is the number of front-edges at $v_i$.

The notion of degeneracy has appeared under other names and many equivalent definitions exist (see e.g.~\cite{ToftJensen}). Note that the degeneracy of a graph is at most $\Delta$. A graph is $1$-degenerate if and only if it is a forest. Planar graphs are $5$-degenerate. Other graph classes of bounded degeneracy include graphs of bounded genus, bounded tree-width, and graphs excluding a fixed minor. 

It is clear that $\Delta\leq \chi' (G)$ for any graph $G$. The celebrated Vizing's Theorem \cite{vizing} states that $\chi'(G)\in\{\Delta, \Delta +1\}$. The following relationship between edge coloring and degeneracy, which is also due to Vizing, will be used extensively in the design of our algorithm. 
\begin{theorem}[Vizing \cite{vizingd, bog}]
\label{vizing}
Let $G$ be a $d$-degenerate graph of maximum degree $\Delta$. If $\Delta\geq 2d$, then $\Delta$ colors suffice for edge coloring $G$.
\end{theorem} 
\subsubsection{Our contribution.}
By Vizing's theorem, there is a trivial upper bound on the advice complexity of \ec of $O(m\log \Delta)$ bits. We improve on this upper bound for $d$-degenerate graphs by showing that $O(m\log d)$ bits of advice suffice to achieve optimality. In particular, only $O(m)$ bits of advice are needed for graphs of bounded degeneracy. The algorithm that we present works in both the advice-on-tape and the advice-with-request model. On the hardness side, we show that $\Omega(m)$ bits of advice are required in order to achieve a competitive ratio better than $2$. This lower bound holds even for forests. Finally, we show that in the advice-with-request model, $\Omega(m\log\Delta)$ bits of advice are necessary to achieve optimality for bipartite graphs.

\subsubsection{Related work.} 
While \ec has not previously been studied in the framework of advice complexity, many other online problems have, see e.g.~\cite{A4, A2, A1, A3, Alist, Aschedule, Ak-server, Avtripartite, Avbipartite, Avpath, Abuffer2, Amakespan, Aexplore, Ais, Vknapsack, Asteiner, Aasg}. We remark that, contrary to \ec, for several of the problems studied in the literature, sublinear advice (in the number of requests) suffice to achieve a competitive ratio better than that of the best deterministic algorithm without advice. This is the case for online problems such as bin-packing~\cite{Abp}, list accessing~\cite{Alist}, knapsack~\cite{Vknapsack}, makespan minimization~\cite{Amakespan}, paging~\cite{A1} and ski-rental\cite{A4}. 

The computational complexity of (offline) edge coloring is well-studied. In general, deciding if $\chi'(G)=\Delta$ is NP-complete~\cite{NP}, but a $(\Delta+1)$-edge-coloring can always be found in polynomial time. For planar graphs, fast algorithms exist for most values of $\Delta$ (see e.g.~\cite{planar1, planar2}). Also, the edge coloring guaranteed to exist by Theorem~\ref{vizing} can be computed efficiently \cite{efficient}.

\section{An algorithm for $d$-degenerate graphs}
In this section, we present the algorithm for $d$-degenerate graphs in the advice-with-request model. As mentioned in the introduction, converting the algorithm to the advice-on-tape model is straightforward.
\begin{theorem}
\label{main}
Let $d\in\mathbb{N}$. For the class of $d$-degenerate graphs, there exists an \ec algorithm which always produces an optimal coloring and uses $$ 1+\ceil{\log (2d)}+\ceil{\log (d+1)}=O(\log d)$$
bits of advice per edge and, hence, $O(m\log d)$ bits of advice in total.
\end{theorem}

Theorem~\ref{main} assumes that the degeneracy of the input graph is at most $d$, where $d$ is a constant hard-wired into both the algorithm and the oracle. In Theorem~\ref{main2}, we show how the assumption that $d$ is constant can be removed by communicating $d$ as part of the advice. Exactly how to do this depends on the model of advice complexity. 

In order to prove Theorem~\ref{main}, we start by assuming that $2d$ divides the maximum degree, $\Delta$, of the input graph. Later on, we will show how to reduce the general case to this special case. 

Let $G=(V,E)$ be a $d$-degenerate input graph of maximum degree $\Delta$ and let $a=\frac{\Delta}{2d}\in\mathbb{N}$. We will first give a high-level description of the oracle and the corresponding algorithm. The main idea is to partition the edges $E$ into $a$ disjoint subsets, $E_1,\ldots , E_a$, such that the maximum degree of the graph $(V,E_j)$ is $2d$ for $1\leq j\leq a$ (this is possible since we are assuming $\Delta=2d\cdot a$). By Theorem~\ref{vizing}, $(E, V_j)$ is $2d$-edge-colorable. Thus, if the algorithm knew how to make such a partition, then, using $O(\log d)$ bits of advice per edge, the algorithm could make an optimal edge coloring of each $E_j$, and hence all of $G$.

However, the oracle cannot afford to compute such a partition and then simply encode the index $j$ such that $e\in E_j$ for each edge, since this would require too much advice per edge if $a$ is large. Instead, the oracle finds a specific partition which is based on the arrival time of the edges  and the fact that $G$ is $d$-degenerate. This partition is such that when an edge $e$ is revealed, the algorithm itself can (without advice) compute a sufficiently small set of indices which always contains the correct index $j$.~This makes it possible to reduce the number of advice bits needed for the algorithm to learn \mbox{the index $j$.}

In order to produce this partition, the oracle orders the vertices of the $d$-degenerate input graph such that no vertex has more than $d$ back-edges. Starting with the first vertex in this ordering, the oracle processes the front-edges of each vertex ordered by (increasing) time of arrival. For each edge, the oracle determines the lowest index $j'$ such that the edge can be assigned to $E_{j'}$ while maintaining that $(V, E_{j'})$ has maximum degree at most $2d$. Note that whenever a front-edge, $e$, at $v$ is being processed, the oracle has already assigned all back-edges of $v$ to some sets in the partition. Since these back-edges may arrive later than $e$, they may be unknown to the algorithm at the time when $e$ is revealed. Therefore, the advice for the front-edge $e$ will warn the algorithm not to assign $e$ to $E_{j}$ if this would later on prevent the intended assignment of some back-edge at $v$ to $E_{j}$.  

\subsubsection{The oracle for the case where $2d$ divides $\Delta$.}
We now give a formal description of the oracle and the algorithm. To each edge $e\in E$, the oracle associates a bit string, $B(e)$, of length $\ceil{\log (2d)}+\ceil{\log (d+1)}$ by following Procedure \ref{euclid}. 
\vspace{-0.2cm}
\begin{algorithm}[h]
\caption{Constructing the advice in the case where $2d$ divides $\Delta$.}\label{euclid}
\begin{algorithmic}[1]
\normalsize
\algrenewcommand\algorithmicrequire{\textbf{Input:}}
\algrenewcommand\algorithmicensure{\textbf{Output:}}
\Require A $d$-degenerate graph $G=(V,E)$ of maximum degree $\Delta$ where \mbox{$a\cdot 2d=\Delta$} for some $a\in\mathbb{N}$, together with arrival times of the edges.
\Ensure A bit string $B(e)$ of length $\ceil{\log 2d}+\ceil{\log (d+1)}$ for each edge $e\in E$.

\State $E_j\gets\emptyset$ for $1\leq j\leq a$
\State Compute an ordering $\{v_1,\ldots , v_n\}$ of the vertices of $G$ such that, for \mbox{$1\leq i\leq n$}, the vertex $v_i$ is adjacent to at most $d$ vertices in $\{v_1,\ldots , v_{i-1}\}$.

\vspace{0.2em}
\State Let $E(v_i)$ denote the edges incident to $v_i\in V$.
\State $\text{Prev}(e, v_i)\gets \{f\in E(v_i) \colon \text{$f$ arrives before $e$}\}$ for $e\in E, v_i\in V$.
\vspace{0.15em}
\For{$i=1$ \textbf{to} $n$}
\State Let $\{e_1,\ldots , e_{d_f(v_i)}\}$ be the front-edges at $v_i$ ordered by time of arrival.
\For{$s=1$ \textbf{to} $d_f(v_i)$}
\State $e \gets e_s$
\State $J(e) \gets \big\{j : \ab{E_j\cap \text{Prev}(e, v_i)}\leq 2d-1\big\}$ \label{je}
\State Let $j'$ be the lowest index such that $\ab{E_{j'}\cap E(v_i)}\leq 2d-1$. \label{ji}
\State $E_{j'} \gets E_{j'}\cup\{e\}$
\State Use the last $\ceil{\log (d+1)}$ bits of $B(e)$ to encode 
$\ab{\{j\in J(e) \colon j<j'\}}$. \label{jia}
\EndFor
\EndFor
\State Compute $2d$-edge-colorings $\mathcal{C}_j$ of $(V,E_j)$ for all $1\leq j\leq a$.
\vspace{0.05em}
\State For each edge $e\in E$, use the first $\ceil{\log(2d)}$ bits of $B(e)$ to encode the color assigned to $e$ in $\mathcal{C}_j$, where $j$ is such that $e\in E_j$.
\end{algorithmic}
\end{algorithm}
\FloatBarrier
In order to prove the correctness of Procedure~\ref{euclid}, we introduce the following terminology: At any point during the execution of Procedure~\ref{euclid}, we say that an edge can \emph{legally} be assigned to a subset $E_j$ if this assignment does not make the maximum degree of $(V,E_j)$ larger than $2d$. Also, we let $\mathcal{P}=\{E_1,\ldots , E_a\}$. We will show in Lemma~\ref{lemma} that the index $j'$ in line \ref{ji} is such that $e$ can legally be assigned to $E_{j'}$ and that the number in line \ref{jia} can be encoded using $\ceil{\log (d+1)}$ bits. 

\begin{lemma} Suppose that during the execution of Procedure~\ref{euclid}, the second for-loop has just been entered and that $e=e_s$. Let $j'$ be the lowest index such that $\ab{E_{j'}\cap E(v_i)}\leq 2d-1$. Then, $e$ can legally be assigned to $E_{j'}$. Furthermore, $j'$ is among the $d+1$ lowest indices in $J(e)$.
\label{lemma}
\end{lemma}
\begin{proof}

Assume that $e=(v_i, v_k)$ is a front-edge at $v_i$ and a back-edge at $v_k$. Because $i<k$, none of the front-edges at $v_k$ has yet been assigned to any subset in $\mathcal{P}$. Since $v_k$ has at most $d$ back-edges (including $e$), it follows that no subset in $\mathcal{P}$ currently contains more than $d-1$ edges incident to $v_k$. Thus, if $e$ cannot legally be assigned to some subset $E_j$, this can only be because it would violate the degree constraint at $v_i$. That is, $e$ can be legally assigned to $E_j$ if and only if $\ab{E_j\cap E(v_i)}\leq 2d-1$. Since at most $\Delta-1=a\cdot 2d-1$ edges incident to $v_i$ have arrived earlier than $e$, and since there are $a$ subsets in $\mathcal{P}$, this implies that $e$ can legally be assigned to at least one subset in $\mathcal{P}$.

Let $j'$ be the lowest index such that $e$ can legally be assigned to $E_{j'}$. Clearly, $j'\in J(e)$ since $\text{Prev}(e,v_i)\subseteq E(v_i)$. We will show that $j'$ is in fact among the $d+1$ lowest indices in $J(e)$. Let $j\in J(e)$. 
By definition of $J(e)$, the number of edges in $E_{j}$ which are incident to $v_i$ and arrive before $e$ is at most $2d-1$. Thus, if $e$ cannot legally be assigned to $E_{j}$, then there must be an edge $f\in E_j$ which is incident to $v_i$ but arrives later than $e$. The front-edges at $v_i$ arriving later than $e$ has not yet been assigned to any subset in $\mathcal{P}$, and so $f$ must be a back-edge at $v_i$. Since $v_i$ has at most $d$ back-edges, there can be at most $d$ indices $j\in J(e)$ such that $e$ cannot legally be assigned to $E_j$. It follows that $j'$ must be among the $d+1$ lowest indices in $J(e)$.
\end{proof}

Combining the assumption that $G$ has maximum degree $\Delta=a 2d$ with Lemma~\ref{lemma} shows that Procedure~\ref{euclid} constructs a partition $E_1,\ldots , E_a$ of $E$ such that the maximum degree of $(V, E_i)$ is $2d$, for $1\leq i\leq n$. Furthermore, the number in line \ref{jia} is at most $d+1$ (and non-negative), and hence it can be encoded in binary using $\ceil{\log (d+1)}$ bits. It follows from Theorem~\ref{vizing} that each of the graphs $(V,E_i)$ can be edge colored using $2d$ colors since they all have maximum degree $2d$ and are $d$-degenerate (because they are subgraphs of $G$ which is $d$-degenerate). This proves the correctness of Procedure~\ref{euclid}.

\subsubsection{The algorithm for the case where $2d$ divides $\Delta$.} We now describe how the algorithm, $\ALG$, uses the advice provided by the oracle. Note that when an edge $e$ arrives, $\ALG$ is able to compute the set of indices $J(e)$ as defined in line \ref{je} of Procedure~\ref{euclid}, since $J(e)$ only depends on $d$ and the edges that have arrived earlier than $e$. Thus, $\ALG$ can compute the index $j'$ such that $e$ was assigned to $E_{j'}$ by Procedure~\ref{euclid} by learning the number $\ab{\{j\in J(e) \colon j<j'\}}$ from the last $\ceil{\log (d+1)}$ bits of $B(e)$. The algorithm (internally) assigns $e$ to $E_{j'}$. 
Then, $\ALG$ reads the integer, $c$, encoded by the first $\ceil{\log (2d)}$ bits of $B(e)$ and colors $e$ with the color $((j'-1)2d+c)$.

It follows that for all $1\leq j\leq a$, the algorithm colors the edges of $E_j$ with the colors $((j-1)2d+1),\ldots , j2d$ and produces a coloring of $E_j$ which is equivalent to the coloring $\mathcal{C}_j$ computed by the oracle. Thus, $\ALG$ produces an optimal edge coloring of $G$.

\subsubsection{The general case.}
Using the algorithm for the case where $2d$ divides $\Delta$ as a subroutine, we are now ready to prove Theorem~\ref{main}.

\begin{proof}[of Theorem~\ref{main}] Fix $d\in\mathbb{N}$. We will describe an algorithm, $\ALG$, and an oracle, \o, satisfying the conditions of the theorem. Let $G=(V,E)$ be a $d$-degenerate input graph of maximum degree $\Delta$. To each edge $e\in E$, the oracle associates a bit string, $B(e)$, of length $1+\ceil{\log (2d)}+\ceil{\log (d+1)}$. 
The definition of $B$ falls into two cases.

\emph{Case: $\Delta<2d$.} The oracle computes an optimal edge coloring $\mathcal{C}$ of $G$. Since $\Delta<2d$, Vizing's Theorem implies that at most $2d$ different colors are used in $\mathcal{C}$. Let $e\in E$. The first bit of $B(e)$ will be a $0$. The next $\lceil \log (2d)\rceil$ bits will encode the color assigned to $e$ in $\mathcal{C}$. The remaining $\ceil {\log d}$ bits of $B(e)$ are set arbitrarily.

\emph{Case: $\Delta\geq 2d$.} Fix $a,b\in\mathbb{N}$ such that $\Delta=a2d+b$ and $0\leq b\leq 2d-1$. By assumption, $a\geq 1$. The oracle computes an optimal edge coloring $\mathcal{C}$ of $G$. Since $\Delta\geq 2d$, Theorem \ref{vizing} implies that $\Delta=a2d+b$ colors are used in $\mathcal{C}$. Let $E_0$ be the edges colored with the colors $1,2,\ldots , b$. For $e\in E_0$, the bit string $B(e)$ is defined as follows: The first bit is a $0$. The next $\lceil \log (2d)\rceil$ bits encode the color assigned to $e$ in $\mathcal{C}$ (this is clearly possible since $b\leq 2d-1$). The remaining $\ceil {\log d}$ bits of $B(e)$ are set arbitrarily.

Let $G'=(V, E\setminus E_0)$. Since $G'$ is $a2d$-edge-colorable, its maximum degree is at most $a2d$. On the other hand, no vertex in $V$ is incident to more than $b$ edges from $E_0$, and so the maximum degree is at least $\Delta-b=a2d$. Furthermore, removal of edges cannot increase the degeneracy of a graph. Thus, $G'$ must be $d$-degenerate. For edges in $G'$, the first bit of $B(e)$ is set to $1$. The remaining bits of $B(e)$ are constructed by running Procedure~\ref{euclid} on $G'$.

We will now define the algorithm, \ALG. For technical reasons, and since the algorithm does not know $\chi'(G)$, we begin by allowing the algorithm to use colors from $\{0,1\}\times\mathbb{N}$. The algorithm receives the advice $B(e)$ along with each edge $e\in E$. If the first bit of $B(e)$ is a $0$, the algorithm learns which color, $c$, to use for $e\in E_0$ by reading the next $\ceil{\log(2d)}$ bits of $B(e)$. It then assigns the color $(0,c)$ to $e$. If the first bit of $B(e)$ is a $1$ then $\ALG$ simulates, using the remaining bits of $B(e)$, the algorithm for the case where $2d$ divides $\Delta$ with $G'$ as input graph. If that algorithm would assign the color $c$ to $e$, \ALG assigns the color $(1,c)$ to $e$. One can easily modify $\ALG$ to use colors from the set $\{1,\ldots , \chi'(G)\}$ as follows: The first time some color $(i, c), i\in\{0,1\},$ is supposed to be used, $\ALG$ selects the lowest color $c'$ from $\{1,\ldots , \chi'(G)\}$ which has not yet been used. From then on, $\ALG$ always uses the color $c'$ instead of $(i,c)$. 
\end{proof}

\subsubsection{Improvements of the algorithm.} The family of algorithms from Theorem~\ref{main} can be used to create a single algorithm which works even if we do not assume that a constant upper bound on the degeneracy is known a priori. In the advice-with-request model, the oracle starts by computing the degeneracy, $\dgn(G)$, of the input graph $G$. Then, the oracle finds the largest integer $d$ such that $1+\ceil{\log (2d)}+\ceil{\log (d+1)}=1+\ceil{\log (2\dgn(G))}+\ceil{\log (\dgn(G)+1)}$. Clearly, $\dgn(G)\leq d$ and hence $G$ is $d$-degenerate. When the algorithm receives the very first advice string, it determines $d$ from the length of the advice received. From there on, Theorem~\ref{main} applies. Note that we do not increase the amount of advice by using $d$ instead of $\dgn(G)$ as an upper bound on the degeneracy. In the advice-on-tape model, the oracle can simply write the value $d$ onto the advice tape in a self-delimiting way using $O(\log d)$ bits (e.g., by writing $\ceil{\log d}$ in unary and then $d$ itself in binary). This gives the following theorem.

\begin{theorem}
\label{main2}
In both the advice-with-request and the advice-on-tape model, there exists an \ec algorithm which produces an optimal coloring and uses $O(m\log d)$ bits of advice in total, where $d$ is the degeneracy of the input graph.
\end{theorem}

We will show in Theorem~\ref{lower} that in the advice-with-request model, at least $\Omega(\log d)$ bits per edge are required to achieve optimality. Note that this matches asymptotically the upper bound of Theorems~\ref{main} and \ref{main2}. However, the exact number of bits used by the algorithm presented can be lowered. For example, we would like to mention that the algorithm can rather easily be modified to use only a single bit per edge in the case of 1-degenerate graphs (forests). 

\section{Lower bounds}
\subsection{Sublinear advice is no better than no advice}
Recall that one very interesting aspect of the advice-on-tape model is that it allows an algorithm to read a sublinear number of advice bits. However, we will now show that linear advice is required to break the lower bound of $2$ on the competitive ratio for \ec. We remark that the hard input instances used in Theorem~\ref{2lower} are the same as the ones used in \cite{edgecoloring} to obtain the lower bound of $2$ for the competitive ratio of algorithms without advice. The proof of Theorem~\ref{2lower} essentially shows how the techniques used in \cite{edgecoloring} can be extended to obtain a lower bound which holds even for algorithms with sublinear advice. 
\begin{theorem}
\label{2lower}
Let $\varepsilon >0$ and let \ALG be a $(2-\varepsilon)$-competitive \ec algorithm. Then \ALG must read at least $\Omega (m)$ bits of advice, where $m$ is the number of edges. This lower bound holds even if the input graph is required to be a forest. The constant of $\Omega(m)$ depends on $\varepsilon$.   \end{theorem}
\begin{proof}
Let \ALG be a $(2-\varepsilon)$ competitive algorithm in the advice-on-tape model (in the advice-with-request model, Theorem~\ref{2lower} follows directly from \cite{edgecoloring}). By definition, there exists a constant $c_0$ such that $\ALG(\sigma)\leq (2-\varepsilon)\OPT(\sigma)+c_0$ for any input sequence $\sigma$. Fix $\Delta\geq 2$ such that $\varepsilon\Delta>c_0+1$. The adversary graph will be a forest of maximum degree $\Delta$ (and therefore $\Delta$-edge-colorable). We introduce some notation which will be used in the proof. Let $\alpha=(\Delta - 1)\cdot \binom{2\Delta-2}{\Delta-1}+1$, let $\beta=\binom{\alpha}{\Delta}$ and let $R\in\mathbb{N}$ be a large integer. A \emph{star} is the complete bipartite graph $K_{1,\Delta-1}$. We say that a collection of stars are \emph{colored the same} if the edges of all of the stars are colored using the same $\Delta-1$ colors. The values $\alpha$ and $\beta$ have been chosen such that the following holds: \nopagebreak

\emph{Fact 1:} Let $\mathcal{C}$ be an edge coloring of $\alpha$ stars using at most $2\Delta-2$ colors. Then, at least $\Delta$ stars must be colored the same in $\mathcal{C}$.

\emph{Fact 2:} Let $\mathcal{C}_1, \mathcal{C}_2, \ldots , \mathcal{C}_k$ be edge colorings of $\alpha$ stars such that each edge coloring uses at most $2\Delta-2$ colors. Then, there exist $\Delta$ stars such that these stars are colored the same in at least $\ceil{k/\beta}$ of the colorings. 

Fact 1 follows from the pigeonhole principle. Fact 2 follows since there are $\beta$ ways to select the $\Delta$ stars guaranteed to exist by Fact 1. 

The total number of edges in the forest will be $m=(\alpha+\Delta)R$. Let $b$ be the maximum number of advice bits read by $\ALG$ on inputs of length $m$. Note that each of the $2^b$ possible advice strings read by $\ALG$ on inputs of length $m$ corresponds to a deterministic online algorithm without advice. Using this observation, we convert $\ALG$ into $2^b$ deterministic algorithms, $\A_1,\ldots, \A_{2^b}$, such that $\min_j \A_j(\sigma)\leq \ALG(\sigma)$ for any sequence $\sigma$ of $m$ edges. We say that $\A_j$ is \emph{alive} if the number of colors used by $\A_j$ so far is at most $2\Delta-2$. 

We will now describe the adversary. The adversary starts by revealing $R$ \emph{rows}, where each row consists of $\alpha$ stars. The remaining edges are revealed in a number of rounds, one for each row.

In round $i$, where $1\leq i\leq R$, the adversary uses the following strategy: Let $k$ be the number of algorithms alive just before round $i$. The adversary selects $\Delta$ stars from row $i$ which are colored the same by at least $\ceil{k/\beta}$ of the algorithms alive. Since an algorithm which is alive has used at most $2\Delta-2$ colors, this is always possible by Fact 2. Let $v_1^i,\ldots , v_{\Delta}^i$ be the vertices of degree $\Delta-1$ in the stars selected. The adversary reveals $\Delta$ edges, $(v, v_1^i),\ldots , (v,v_\Delta^i)$, connecting these vertices to a new vertex, $v$. An algorithm $\A_j$ which have colored the selected stars with the same $\Delta-1$ colors is forced to use $\Delta$ other colors for these new edges and, hence, to use $2\Delta-1$ colors in total. Thus, at the end of round $i$, the number of algorithms alive is at most $k-\ceil{k/\beta}\leq \left(1-1/\beta\right)k$.

Since $\varepsilon \Delta>c_0+1$, we have that $2\Delta-1 > (2-\varepsilon)\Delta+c_0.$ In particular, at least one algorithm $\A_j$ must be alive after round $R$ since the number of colors used by $\ALG$ is at most $(2-\varepsilon)\Delta+c_0$. Before the first round, the number of algorithms alive is at most $2^b$. After the last round, the number of algorithms alive is therefore at most $\left(1-1/\beta\right)^{R}2^b$. Thus, it must hold that $1\leq \left(1-1/\beta\right)^{R}2^b$. But, this implies that
\begin{equation}
b\geq  R\log\left(\frac{\beta}{\beta-1}\right)=\frac{\log\left(\frac{\beta}{\beta-1}\right)}{\alpha+\Delta}m=\Omega(m).
\end{equation}
This proves the theorem since the number of rounds $R$ (and therefore also $m$) can be chosen arbitrarily large.
\end{proof}

We remark that the hidden constant in the $\Omega(m)$ lower bound of Theorem~\ref{2lower} decreases very fast as $\varepsilon$ tends to zero. However, the main point of Theorem~\ref{2lower}, that sublinear advice does not offer any advantage, is not affected by this.

\subsection{Tight lower bounds in the advice-with-request model}
As we have shown, $O(m)$ bits of advice suffice to achieve optimality for graphs of bounded degeneracy. A natural question is whether this holds for general graphs. We give a partial negative answer by showing that in the advice-with-request model, this is not the case, not even for bipartite graphs. 

It is a well-known result of K\"onig \cite{Konig} that bipartite graphs are $\Delta$-edge-colorable. 
In the proof of Theorem~\ref{lower}, we will use the following gadget to ensure that two edges cannot be assigned different colors in an optimal edge coloring.

\begin{definition}
Let $n\geq 1$. The graph $H_n$ consists of a complete bipartite graph $K_{n, n}=(L,R)$ together with vertices $v_l, v_r$ and edges \mbox{$\{(v_l, v) : v\in L\}$}, $\{(v_r, v) : v\in R\}$. The vertex $v_l$ ($v_r$) is denoted the leftmost (rightmost) vertex.
\end{definition}
We say that two edges $e_1=(x_1, y_1)$ and $e_2=(x_2, y_2)$ are \emph{connected by an $H_n$} if $y_1$ is the leftmost vertex and $y_2$ is the rightmost vertex of the same $H_n$ (and neither $x_1$ nor $x_2$ is part of that $H_n$). See Figure~\ref{graf}.
\vspace{-0.1cm}
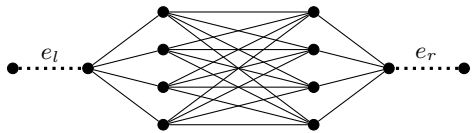
\begin{figure}[h!]
\begin{center}
\begin{tikzpicture}[every node/.style=defnode]

		\node   (0) at (-3, -0) {};
		\node   (1) at (-2, -0) {};
		\node   (2) at (-1, 0.25) {};
		\node   (3) at (-1, 0.75) {};
		\node   (4) at (-1, -0.25) {};
		\node   (5) at (-1, -0.75) {};
		\node   (6) at (1, -0.25) {};
		\node   (7) at (1, 0.75) {};
		\node   (8) at (2, -0) {};
		\node   (9) at (1, 0.25) {};
		\node   (10) at (3, -0) {};
		\node   (11) at (1, -0.75) {};

		\path (0) edge [dotted, very thick] node[fill=none,draw=none,above] {$e_l$}(1);
		\path (1) edge [] (3);
		\path (1) edge [] (2);
		\path (1) edge [] (4);
		\path (1) edge [] (5);
		\path (10) edge [dotted, very thick] node[fill=none,draw=none,above] {$e_r$}(8);
		\path (8) edge [] (7);
		\path (8) edge [] (9);
		\path (8) edge [] (6);
		\path (8) edge [] (11);
		\path (3) edge [] (7);
		\path (3) edge [] (9);
		\path (3) edge [] (6);
		\path (3) edge [] (11);
		\path (2) edge [] (7);
		\path (2) edge [] (9);
		\path (2) edge [] (6);
		\path (2) edge [] (11);
		\path (4) edge [] (7);
		\path (4) edge [] (9);
		\path (4) edge [] (6);
		\path (4) edge [] (11);
		\path (5) edge [] (7);
		\path (5) edge [] (9);
		\path (5) edge [] (6);
		\path (5) edge [] (11);
\end{tikzpicture}
\end{center}
\caption{$G_4$: Two edges $e_l$ and $e_r$ (dashed lines) connected by an $H_4$ (solid lines).}
\label{graf}
\end{figure}
\vspace{-0.4cm}
\begin{lemma}
For $n\geq 1$, let $G_n$ be the graph consisting of two edges, $e_l$ and $e_r$, connected by an $H_n$. Then, $G_n$ is $(n+1)$-edge-colorable. On the other hand, an edge coloring of $G_n$ in which $e_l$ and $e_r$ are assigned different colors must use at least $n+2$ colors.
\label{Hlemma}
\end{lemma}
\begin{proof}
$G_n$ can be edge colored using $n+1$ colors since it is a bipartite graph of maximum degree $n+1$. Let $\mathcal{C}$ be an edge coloring of $G_n$ such that $C(e_l)\neq C(e_r)$, where $C(e)$ is the color assigned to the edge $e$. Suppose, by way of contradiction, that only $n+1$ different colors are used in $\mathcal{C}$. Since $e_l$ and $e_r$ are colored differently, the set of colors used for edges between $v_l$ and $L$ cannot be identical to the set of colors used for edges between $v_r$ and $R$, since this would contradict that $\mathcal{C}$ uses only $n+1$ colors. Thus, there exists a color, $c$, such that there is an edge $e=(v_l, v), v\in L$ colored with the color $c$, while no edge between $v_r$ and $R$ is colored with the color $c$. It follows that for each $u\in R$, there must be an edge between $u$ and a vertex in $L$ colored with the color $c$, since $u$ has degree $n+1$, $\mathcal{C}$ uses $n+1$ colors and the edge $(v_r,u)$ is not colored with the color $c$. In particular, since $\ab{L}=\ab{R}$, there must be an edge from a vertex in $R$ to $v$ colored with the color $c$. This is a contradiction, since $(v_l, v)$ is also colored with the color $c$.
\end{proof}

\begin{theorem}
\label{lower}
An optimal \ec algorithm in the advice-with-request model must use at least $\Omega (\log \Delta)$ bits of advice per edge, even for bipartite graphs, where $\Delta$ is the maximum degree of the input graph.
\end{theorem}
\begin{proof}

Fix $\Delta\geq 2$. At the beginning, the adversary reveals two stars $K_{1,\Delta}$. Let $\{x_1,\ldots , x_\Delta\}$ and $\{y_1,\ldots, y_\Delta\}$ be the vertices of degree $1$ in each of these two stars, and let $x$ and $y$ be the center vertices. Furthermore, let $t$ be the time step right after both stars have been revealed. At time $t$, the adversary picks a permutation $\pi$ of $\{1,2,\ldots , \Delta\}$. For each $1\leq i\leq\Delta$, the edge $(x,x_i)$ is then connected to the edge $(y,y_{\pi(i)})$ by an $H_{\Delta-1}$ through $x_i$ and $y_{\pi(i)}$. Since the resulting graph is bipartite and has maximum degree $\Delta$, it can be colored using $\Delta$ colors.

Let $\ALG$ be an algorithm in the advice-with-request model such that, at time $t$, the total number of advice bits received by $\ALG$ is strictly less than $\log (\Delta !)$. We claim that $\ALG$ cannot be optimal.  Note that the adversary has $\Delta !$ different permutations to choose from. This implies that there must exist two different permutations $\pi, \pi'$ such that up until time $t$, the algorithm receives exactly the same bits of advice for both of these permutations. Thus, \ALG produces the same coloring, $\mathcal{C}$, of the two stars no matter which of $\pi$ and $\pi'$ the adversary chooses to use. Let $C(u,v)$ be the color assigned to the edge $e=(u,v)$ in $\mathcal{C}$. 
Fix $i$ such that $\pi(i)\neq \pi '(i)$. Because the edges are adjacent, $C(y,y_{\pi(i)})\neq C(y,y_{\pi'(i)})$. Since $C(x, x_i)$ cannot be the same as both $C(y, y_{\pi (i)})$ and $C(y, y_{\pi' (y)})$, we may assume without loss of generality that $C(x,x_i)\neq C(y,y_{\pi(i)})$. By Lemma~\ref{Hlemma}, this implies that $\ALG$ will use at least $\Delta+1$ colors when the adversary chooses the permutation $\pi$. 
We conclude that an optimal algorithm must have received at least $\log (\Delta !)$ bits of advice at time $t$. Since only $2\Delta$ edges are revealed before time $t$, this is only possible if the algorithm receives at least $\frac{\log (\Delta !)}{2\Delta}=\Omega(\log \Delta)$ bits of advice per edge.
\end{proof}

For the adversary graph used in Theorem \ref{lower}, the number of edges is $m=O(\Delta ^3)$. Thus, we may restate the lower bound of $\Omega (\log \Delta)$ bits per edge in terms of $m$ and get a lower bound of $\Omega(\log (m^{1/3}))=\Omega (\log m)$ bits per edge. This shows that even if we insist on measuring the amount of advice solely as a function of $m$ (and not also $\Delta$), the trivial upper bound of $O(m\log m)$ on the advice complexity is still asymptotically tight.  Furthermore, the graph is $\Delta$-regular. It follows that the degeneracy $d$ of the graph is $d=\Delta$. Thus, the lower bound may also be stated in terms of the degeneracy as $\Omega (\log d)$ bits per edge. This matches asymptotically the upper bound of Theorems~\ref{main} and \ref{main2}.

\section{Concluding remarks and open problems}
\label{conclusion}

As a consequence of Euler's formula, the degeneracy of a planar graph is at most $5$. Thus, Theorem~\ref{main} implies that $8$ bits of advice per edge (and hence $8m=O(m)$ bits in total) suffice to achieve optimality for planar graphs. On the other hand, since a forest is a planar graph, Theorem~\ref{2lower} shows that $\Omega(m)$ bits of advice are necessary just to achieve a competitive ratio better than $2$. As mentioned, the greedy algorithm is $2$-competitive~\cite{edgecoloring} and uses no advice at all. Thus, Theorems~\ref{main} and \ref{2lower} completely determines (asymptotically) the advice complexity of edge coloring planar graphs, in both models of advice complexity.

For bipartite graphs, the picture is not as clear. The lower bound of Theorem~\ref{lower} for bipartite graphs relies on the assumption that an algorithm receives a fixed number of advice bits per edge, and so it only holds in the advice-with-request model. The lower bound may be viewed as a worst-case lower bound: We show that there exist some edges for which $\Omega(\log \Delta)$ bits of advice are required. The advice-on-tape model allows us to also study the amortized number of advice bits per edge. Determining the advice complexity of \ec for bipartite graphs in the advice-on-tape model is left as an interesting open problem.

\subsubsection{Acknowledgements.}
The author wishes to thank Joan Boyar and Lene M. Favrholdt for helpful discussions.

\bibliographystyle{plain}
\bibliography{ref}

\end{document}